\documentclass[11pt]{article}

\usepackage{amsthm,amsmath,amsfonts, amssymb}
\usepackage{mathtools}
\usepackage{fullpage}
\usepackage[hidelinks]{hyperref}

\newtheorem{lemma}{Lemma}
\newtheorem{theorem}{Theorem}
\newtheorem{corollary}{Corollary}
\newtheorem{definition}{Definition}

\newtheorem{conjecture}{Conjecture}

\newcommand{\F}{\mathbb{F}}

\newcommand{\col}{\mathrm{Col}}
\newcommand{\row}{\mathrm{Row}}

\newcommand{\eps}{\varepsilon}
\newcommand{\rk}{\mathsf{rk}}
\newcommand{\mrk}{\mathsf{minrk_{\F}}}
\newcommand{\E}{\mathbb{E}}
\newcommand{\mrktwo}{\mathsf{minrk_{2}}}

\newcommand{\cG}{\mathcal{G}}

\newcommand{\cM}{\mathcal{M}}

\newcommand{\cK}{\mathcal{K}}

\begin{document}
\begin{titlepage}
\title{\bf The Minrank of Random Graphs}
\author{
Alexander Golovnev\thanks{Courant Institute of Mathematical Sciences, New York
 University.}~\thanks{Supported by the Simons Collaboration on Algorithms and Geometry and by the National Science Foundation (NSF) under Grant No.~CCF-1320188. Any opinions, findings, and conclusions or recommendations expressed in this material are those of the authors and do not necessarily reflect the views of the NSF.}
\and
Oded Regev\footnotemark[1]~\footnotemark[2]
\and
Omri Weinstein\thanks{Columbia University}~\thanks{The work was done when the author was supported by a Simons Junior fellowship.}
}
\date{}
\maketitle
\thispagestyle{empty}

\begin{abstract}
The \emph{minrank} of a graph $G$ %
is the minimum rank of a matrix $M$ that can be obtained from the adjacency matrix of $G$ by
 switching some ones to zeros (i.e., deleting edges) and then setting all diagonal entries to one. 
This quantity is closely related to the fundamental information-theoretic problems of (linear) \emph{index coding} 
(Bar-Yossef et al., FOCS'06), %
network coding and distributed storage, and to Valiant's approach for proving superlinear circuit lower bounds (Valiant, Boolean Function Complexity '92). 

We prove tight bounds on the minrank of random Erd{\H o}s-R{\'e}nyi graphs $G(n,p)$ for all regimes of $p\in[0,1]$. In particular, for any constant $p$, we show that $\mathsf{minrk}(G) = \Theta(n/\log n)$
with high probability, where $G$ is chosen from $G(n,p)$. 
This bound gives a near quadratic improvement over the previous best lower bound of $\Omega(\sqrt{n})$ (Haviv and Langberg, ISIT'12), 
and partially settles an open problem raised by Lubetzky and Stav (FOCS '07). 
Our lower bound matches the well-known upper bound obtained by the ``clique covering" 
solution, and settles the linear index coding problem for random graphs. 

Finally, our result suggests a new avenue of attack, via derandomization, on  
Valiant's approach for proving superlinear lower bounds for logarithmic-depth semilinear circuits.
 
\end{abstract}
\end{titlepage}

\section{Introduction}

In information theory, the \emph{index coding} problem~\cite{BK98a,BBJK06} is the following: 
A sender wishes to \emph{broadcast} over a noiseless channel an $n$-symbol string $x\in \mathbb{F}^n$ to a group of $n$ receivers $R_1,\ldots,R_n$, 
each equipped with some \emph{side information}, namely, a subvector $x_{K_i}$ of $x$ indexed by a subset ${K_i} \subseteq \{x_1,\ldots ,x_n\}$.  
The index coding problem asks what is the minimum length $m$ of a broadcast message that allows each receiver $R_i$ to 
retrieve the $i$th symbol $x_i$, given his side-information $x_{K_i}$ and the broadcasted message.   
The side information of the receivers can be modeled by a directed graph $\cK_n$, 
in which $R_i$ observes %
the symbols $K_i := \{ x_j \; : \; (i,j) \in E(\cK_n) \}$.  $\cK_n$ is sometimes called the \emph{knowledge graph}. 
A canonical example is where $\cK_n$ is the complete graph (with no self-loops) on the vertex set $[n]$, i.e.,  
each receiver observes all but his own symbol. In this simple case, 
broadcasting the sum $\sum_{i=1}^n x_i$ (in $\mathbb{F}$) allows each receiver to retrieve his own symbol, hence $m= 1$. 

This problem is motivated by applications to distributed storage~\cite{AK15b},  
on-demand video streaming (ISCOD, \cite{BirkK06}) and wireless networks (see, e.g., \cite{YZ99a}), where a typical scenario is that 
clients miss information during transmissions of the network, and the network is interested in minimizing the retransmission length by exploiting the 
side information clients already possess.  
In theoretical computer science, index coding is related to some important %
communication models and problems in which players have overlapping information, such as 
the \emph{one-way} communication complexity of the index 
function~\cite{KNR01} and the more general problem of \emph{network coding}~\cite{Ashwede,ERL15}. 
Index coding can also be viewed as an interesting special case of nondeterministic computation in the 
(notoriously difficult to understand) multiparty \emph{Number-On-Forehead} model, which in turn is a promising approach for proving 
data structure and circuit lower bounds~\cite{Patrascu10, pudlak1997boolean, JuknaS11}.  
The minimum length of an index code for a given graph has well-known relations to other important graph parameters.
For instance, it is bounded from below by the size of the maximum independent set, and it is bounded from above 
by the clique-cover number ($\chi(\bar{G})$) since for every clique in $G$, it suffices to 
broadcast a single symbol (recall the example above). The aforementioned connections also led to algorithmic connections (via convex 
relaxations) between the computational complexity of graph coloring and that of computing the minimum index code length of 
a graph \cite{ChlamtacH14}.

In the context of circuit lower bounds, Riis \cite{R07} observed that a certain index coding problem is equivalent to the so-called 
\emph{shift conjecture} of Valiant \cite{V92} (see Subsection \ref{subsec_derendomization} below).
If true, %
this conjecture would resolve a major open problem of proving superlinear 
size lower bound for logarithmic-depth circuits. 

When the encoding function of the index code is \emph{linear} in $x$ (as in the example above), the corresponding scheme is called a \emph{linear index code}. 
In their seminal paper, Bar-Yossef et al. \cite{BBJK06} showed that the minimum length $m$ of a \emph{linear} index code is characterized 
precisely by a parameter of the knowledge graph $\cK_n$, called the \emph{minrank} ($\mrk(\cK_n)$), first introduced by 
Haemers \cite{Haemers79} in the context of Shannon capacity of graphs.%
\footnote{To be precise, this holds only for graphs without self-loops. We will ignore this minor issue in this paper as
it will not affect any of our results.}
Namely,  
$\mrk(\cK_n)$ is the minimum rank (over $\mathbb{F}$) of an $n\times n$ matrix $M$
that ``represents'' $\cK_n$. By ``represents'' we mean a matrix $M$ that 
contains a zero in all entries corresponding to \emph{non-edges}, and non-zero 
entries on the diagonal. Entries corresponding to edges are arbitrary. (Over $\mathbb{F}_2$ 
this is equivalent to being the adjacency matrix of a subgraph of $\cK_n$, with diagonal
entries set to one.)
Note that without the ``diagonal constraint", the above 
minimum would trivially be $0$, and indeed this constraint is what makes the problem interesting and hard to analyze. 
While  linear index codes are in fact optimal for a large class of knowledge graphs (including directed acyclic graphs, perfect graphs, odd ``holes" 
and odd ``anti-holes" %
\cite{BBJK06}), 
there are examples where non-linear codes outperform their linear counterparts~\cite{LU07}. 
In the same paper, Lubetzky and Stav \cite{LU07} posed the following question about  \emph{typical}  knowledge graphs, namely, 

\begin{quote}
\emph{What is the minimum length of an index code for a random knowledge graph 
$\cK_n=\cG_{n,p}$? } %
\end{quote}

\noindent Here, $\cG_{n,p}$ denotes a random Erd{\H o}s-R{\'e}nyi directed graph, 
i.e., a graph on $n$ vertices in which each arc is taken independently with probability $p$.
In this paper, we partially answer this open problem by determining the optimal length of \emph{linear} index codes for such graphs.
In other words, we prove a tight lower bound on the minrank of $\cG_{n,p}$ for all values of $p\in[0,1]$. In particular, 

\begin{theorem}[Main theorem, informal] \label{thm_main_informal}
For any constant $0<p<1$ and any field $\F$ of cardinality $|\F|<n^{O(1)}$, it holds with high probability that 
\[
\mrk(\cG_{n,p}) =  \Theta\left( \frac{n}{\log n} \right) \; . %
\] 
\end{theorem}

The formal quantitative statement of our result can be found in Corollary \ref{cor_main_thm_const_p} below. We note that our general result 
(see Theorem \ref{thm_mrk_gnp}) extends beyond the constant regime to \emph{subconstant} values of $p$, and this feature of our lower bound 
is crucial for potential applications of our result to circuit lower bounds (we elaborate on this in the next subsection).  
Theorem \ref{thm_main_informal} gives a near quadratic improvement over the previously best lower bound of $\Omega(\sqrt{n})$~\cite{LU07,HL12}, 
 and settles the linear index coding problem for random knowledge graphs, as an $O_p(n/\log n)$ linear index coding 
scheme is achievable via the clique-covering solution (see Section \ref{sec_proof_overview}). 

In the following subsection, we propose a concrete (yet admittedly still quite challenging) approach for proving superlinear circuit lower bounds based on a 
potential ``derandomization" of Theorem \ref{thm_main_informal}. 

\subsection{Connections to circuit lower bounds for semilinear circuits} \label{subsec_derendomization}

\paragraph{General log-depth circuits.}
In his seminal line of work, Valiant \cite{V77,V83,V92} proposed a path for proving superlinear  lower bounds on the size of %
circuits with logarithmic depth, one of the main open questions in circuit complexity. 
Informally speaking, Valiant's ``depth reduction" method~\cite{V77,V09} allows, for any constant $\eps$, 
to reduce any circuit of size $O(n)$ and depth 
$O(\log n)$ (with $n$ inputs and $n$ outputs), to a new circuit %
with the same inputs and outputs, where now each output gate is an (arbitrary) Boolean function of $(i)$ at most $n^{\eps}$ 
inputs which are ``hard-wired" to this output gate, and $(ii)$ an additional fixed set of $m=O_{\eps}(n/\log\log n )$ 
``common bits" $b_1(x),\ldots, b_m(x)$ which in general may be arbitrary Boolean functions of the input $x=x_1,\ldots, x_n$. 
Therefore, if one could exhibit a function that cannot be computed in this model using $O(n/\log\log n)$ common bits, this would imply 
a superlinear circuit lower bound for logarithmic depth circuits. 

Valiant~\cite{V92} proposed a concrete candidate hard function for this new model,
namely the function whose input is an $n$-bit string $x$ and a number $i \in \{0,\ldots,n-1\}$
and whose output is the $i$th cyclic shift of $x$. 
Valiant conjectured that no ``pre-wired" circuit as above 
can realize \emph{all} $n$ cyclic shifts using $m=O(n/\log\log n)$ common bits (in fact, Valiant postulated that $m=\Omega(n)$ common bits are required, 
and this still seems plausible). This conjecture is sometimes referred to as \emph{Valiant's shift conjecture}.
As noted earlier in the introduction, Riis \cite{R07} observed that a certain index coding problem is equivalent to this conjecture. 
Let $G=(V,A)$ be a directed graph, and $i\in\{0,\ldots,n-1\}$. We denote by $G^i$ the graph with vertex set $V$ and arc set $A^i=\{(u,v+i (\bmod~n)): (u,v)\in A\}$.
Riis \cite{R07} showed that the following conjecture is equivalent to Valiant's shift conjecture:
\begin{conjecture}\label{conj_riis}
There exists $\eps>0$ such that for all sufficiently large $n$ and every graph $G$ on $n$ vertices with max-out-degree at most $n^{\eps}$, there exists a shift $i$ such that the minimum length of an index coding scheme for $G^i$ (over $\mathbb{F}_2$) is $\omega(n/\log\log n)$.
\end{conjecture}

\paragraph{Semilinear log-depth circuits.}
Let us consider a function $f(x,p)$ whose input is partitioned into two parts, $x\in\{0,1\}^k$ and $p\in\{0,1\}^t$. 
We say that the function $f$ is \emph{semilinear} if for every fixed value of $p=p_0$, the function $f(x,p_0)$ is a linear function (over $\mathbb{F}_2$) of  $x$. 
The class of semilinear functions is quite rich, and includes for instance bilinear functions in $x$ and $p$ (such as matrix multiplication) and permutations $\pi_p(x)$ of $x$ that may depend arbitrarily on $p$.
A circuit $G$ is called \emph{semilinear} if for every fixed value of $p=p_0$, one can assign linear functions to the gates of $G$, so that $G$ computes $f(x,p_0)$.
So it is only the circuit's topology that is fixed, and the linear functions computed by the gates may depend arbitrarily on $p$. 

It is easy to see that a semilinear function with a one-bit output can always be computed by a linear-size log-depth semilinear circuit (namely, the full binary tree).
However, if we consider semilinear functions with $O(n)$ output bits, then the semilinear circuit complexity of a random function is $\Omega(n^2/\log{n})$ with high probability. It is an open problem to prove a superlinear lower bound against log-depth semilinear circuits~\cite{pudlak1997boolean}. 
This would follow from the semilinear variant of Valiant's shift conjecture,
which is equivalent to the following slight modification of Conjecture~\ref{conj_riis}~\cite{pudlak1997boolean,R07}.

\begin{conjecture}\label{conj_riis_semilinear}
There exists $\eps>0$ such that for all sufficiently large $n$ and every graph $G$ on $n$ vertices with max-out-degree at most $n^{\eps}$, there exists a shift $i$ such that the minimum length of a \emph{linear} index coding scheme for $G^i$ (over $\mathbb{F}_2$) is $\omega(n/\log\log n)$. Equivalently,
\[ \forall \; G \text{ of out-degrees at most $n^{\eps}$} \;\;  \exists \;\; i\in [n] \;\;  \mrktwo(G^i) = \omega(n/\log\log n) \; .\] 
\end{conjecture}

Theorem~\ref{thm_main_informal} (and the more precise concentration bound we prove in Theorem \ref{thm_mrk_gnp}) 
asserts that %
with high probability, 
a graph chosen from $\cG_{n,p}$ 
(with $p=n^{\eps-1}$ for the 
expected degree of each vertex to be $n^{\eps}$)
has minrank $\Omega(n)$.  
Conjecture~\ref{conj_riis_semilinear} would follow from a ``derandomization'' of Theorem~\ref{thm_main_informal} 
in which we replace the distribution $\cG_{n,p}$ with a random shift of an arbitrary given graph of the right degree. 
In fact, for the purpose of circuit lower bounds, one could replace cyclic shifts with any (efficiently computable) set of at most $\exp(O(n))$ permutations.
(Since the permutation itself is part of the input, its description size must be linear in $n$.)

\paragraph{Semilinear series-parallel circuits.}
Finally, we mention one last circuit class for which the above ``derandomization" approach might be easier.
Here we replace the depth restriction by another restriction on the topology of the circuit. 
Namely, a circuit $G=(V,A)$ is called \emph{Valiant series-parallel (VSP)}, if there is a labeling of its vertices $l\colon V \to \mathbb{R}$, such that for every arc $(u,v)\in A$, $l(u)<l(v)$, but there is no pair of arcs $(u,v),(u',v')\in A$, such that $l(u)<l(u')<l(v)<l(v')$. Most of the known circuit constructions (i.e., circuit upper bounds) are VSP circuits. 
Thus, it is also a big open question in circuit complexity to prove a superlinear lower bound on the size of semilinear VSP circuits (of arbitrary depth). 

Valiant~\cite{V77}, Calabro~\cite{C08}, and Riis~\cite{R07} show that in order to prove a superlinear lower bound for semilinear VSP circuits, 
it suffices to show that for a sufficiently large \emph{constant} $d$, for every graph $G$ of max-out-degree at most $d$, the minrank of one of its shifts is at least $n/100$. We note that Theorem~\ref{thm_main_informal} for this regime of $p=d/n$ gives a lower bound of $n/20$. Thus, derandomization of the theorem in this regime would imply a superlinear lower bound. Note that in the case of $p=O(n^{-1})$, the entropy of a random graph is only $O(n\log{n})$ bits, hence, information-theoretically it seems easier to derandomize than the case of $p=n^{\eps-1}$.

\subsection{Proof overview of Theorem \ref{thm_main_informal}} \label{sec_proof_overview}
In \cite{LU07}, Lubetzky and Stav showed that for any field $\F$ and a directed graph $G$, 
$$\mrk{(G)}\cdot\mrk{(\bar{G})} \geq n \; .$$
This inequality gives a lower bound of $\Omega(\sqrt{n})$ on the expected value of the minrank of $\cG_{n,1/2}$.
(Indeed, the random variables $\cG_{n,1/2}$ and $\bar{\cG}_{n,1/2}$ have identical distributions).  
Since $\mrk(\cG_{n,p})$  is monotonically non-increasing in $p$, the same bound holds for any $p\leq 1/2$. 
Haviv and Langberg~\cite{HL12} improved this result by proving a lower bound of $\Omega(\sqrt{n})$ 
for all constant $p$ (and not just $p \le 1/2$), and also by showing that the bound holds with high probability.

We now outline the main ideas of our proof. For simplicity we assume that $\F=\F_2$ and $p=1/2$.
To prove that  $\mathsf{minrk}_2(\cG_{n,p}) \geq k $, we need to show that with high probability, $\cG_{n,p}$ has no representing matrix (in the sense of Definition \ref{def_mrk}) 
whose rank is less than $k$. 

As a first attempt, we can show that any \emph{fixed} matrix $M$ with $1$s on the diagonal of rank less than $k$ has very low probability of 
representing a random graph in $\cG_{n,p}$, and then apply a union bound over all such matrices $M$. 
Notice that this probability is simply $2^{-s+n}$, where $s$ is the sparsity of $M$ (i.e., the number of non-zero entries) and the $n$ is to account for the diagonal entries.
Moreover, we observe that the sparsity $s$ of any rank-$k$ matrix 
with $1$s on its main diagonal must be\footnote{To see why, notice that any maximal linearly independent 
set of columns must 
``cover'' all coordinates, i.e., there must not be any coordinate that is zero in all vectors, 
as otherwise we could take the column vector corresponding to that coordinate and it would be linearly independent 
of our set (due to the nonzero diagonal) in contradiction to maximality. 
Assuming all columns have roughly the same number of 1s, we  
obtain that each column has at least $n/k$ 1s, leading to the claimed bound. 
See Lemma~\ref{lem_rank_vs_sparsity_diagonal} for the full proof.}
at least $\approx n^2/k$.
Finally, since the number of $n\times n$ matrices of rank $k$ is $\approx 2^{2nk}$ (as a rank-$k$ matrix can be written as a product of $n\times k$ 
by $k\times n$ matrices, which requires $2nk$ bits to specify), 
by a union bound, the probability that 
$\cG_{n,p}$ contains a subgraph of rank $<k$ is bounded from above by  (roughly)  $2^{2nk}\cdot (1/2)^{n^2/k}$, which is $\ll 1$ for $k=O(\sqrt{n})$.
This recovers the previous 
$\Omega(\sqrt{n})$ lower bound of \cite{HL12} (for all constant $p$, albeit with a much weaker concentration bound). 

To see why this argument is ``stuck'' at $\sqrt{n}$, 
we observe that we are not overcounting and indeed, 
there are $2^{n^{3/2}}$ matrices 
of rank $k \approx n^{1/2}$ and sparsity $s \approx n^{3/2}$.
For instance, we can take the rank $n^{1/2}$ matrix that consists of $n^{1/2}$ diagonal $n^{1/2} \times n^{1/2}$ blocks of $1$s (a disjoint union of $n^{1/2}$ equal-sized cliques), and replace the first $n^{1/2}$ columns with arbitrary values. 
Each such matrix has probability $2^{-n^{3/2}}$ of representing $\cG_{n,p}$ (because of its sparsity) and 
there are $2^{n^{3/2}}$ of them, so the union bound breaks for $k=\Omega(\sqrt{n})$. 

In order to go beyond $\sqrt{n}$, we need two main ideas. 
To illustrate the first idea, notice that in the above example, even though individually each matrix has probability 
$2^{-n^{3/2}}$ of representing $\cG_{n,p}$, these ``bad events'' are highly correlated.
In particular, each of these events implies that $\cG_{n,p}$ must contain
$n^{1/2}-1$ disjoint cliques, an event that happens with roughly
the same probability $2^{-n^{3/2}}$. Therefore, we see that the probability that the \emph{union} of these bad events 
happens is only $2^{-n^{3/2}}$, greatly improving on the naive union bound argument. 
(We remark that this idea of ``bunching together related events'' is reminiscent of the chaining technique as used, e.g., in analyzing
Gaussian processes.)
More generally, the first idea (and also centerpiece) of our proof is Lemma~\ref{lem_k'_n'_n_k}, which shows 
that every matrix must %
contain a ``nice'' submatrix (in a sense to be defined below).
The second and final idea, described in the next paragraph, will be to bound the number of ``nice''
submatrices, from which the proof would follow by a union bound over all such submatrices. 

Before defining what we mean by ``nice'', we mention the following elementary yet crucial fact in our proof:   
Every rank $k$ matrix is uniquely determined by specifying some $k$ linearly independent rows, 
and some $k$ linearly independent columns (i.e., a row basis and a column basis) including the
indices of these rows and columns (see Lemma~\ref{lem_basis_determines_matrix}). 
This lemma implies that we can encode a matrix using only 
$\approx s_{basis}\cdot \log n$ bits, where $s_{basis}$ is the 
minimal sparsity of a pair of row and column bases that are guaranteed to exist.
This in turn implies that there are  only $\approx 2^{s_{basis} \log n}$ such matrices.
Now, since the average number of $1$s in a row or in a column of a matrix of sparsity $s$ 
is $s/n$, one might hope that such a matrix contains a pair of row and column bases 
of sparsity $k\cdot (s/n)$, and this is precisely our definition of a ``nice'' matrix.
(Obviously, not all matrices are nice, and as the previous example shows, there are lots of ``unbalanced'' matrices 
where the nonzero entries are all concentrated on a small number of columns, hence they have no sparse column 
basis even though the average sparsity of a column is very low; this is exactly why we need
to go to submatrices.)

To complete this overview, notice that using the bound on the number of ``nice''
matrices, the union bound yields
\[ 2^{ks\log(n)/n}\cdot (1/2)^{s}, \] 
so one could set the rank parameter $k$ to be as large as $\Theta(n/\log n)$ and the above expression would still be $\ll 1$. 
A similar bound holds for nice submatrices, completing the proof.

\section{Preliminaries}
For an integer $n$, we denote the set $\{1,\ldots,n\}$ by $[n]$. For an integer $n$ and $0\le p\le 1$, we denote  by 
$\cG_{n,p}$ the probability space over the directed graphs on $n$ vertices where each arc is taken independently with probability $p$.

For a directed graph $G$, we denote by $\chi(G)$ the chromatic number of the undirected graph that has the same set of vertices as $G$, and an edge in place of every arc of $G$. By $\bar{G}$ we mean a directed graph on the same set of vertices as $G$ that contains an arc if and only if $G$ does not contain it.\footnote{Throughout the paper we assume that graphs under consideration do not contain self-loops. In particular, neither $G$ nor $\bar{G}$ has self-loops.}

Let $\F$ be a finite field. For a vector $v\in \F^n$, we denote by $v^{j}$ the $j$th entry of $v$,
and by $v^{\leq j} \in \F^j$ the vector $v$ truncated to its first $j$ coordinates. For a matrix $M\in\F^{n\times n}$ and indices $i,j\in[n]$, let $M_{i,j}$ 
be the entry in the $i$th row and $j$th column of $M, \col_i(M)$ 
be the $i$th column of $M$, $\row_i(M)$ be the $i$th row of $M$, and $\rk(M)$ be the rank of $M$ over $\F$. 

By a \emph{principal submatrix} we mean a submatrix whose set of row indices is the same as the set of column indices.
By the \emph{leading principal submatrix} of size $k$ we mean a principal submatrix that contains the first $k$ columns and rows.

For a matrix $M\in\F^{n \times n}$, the sparsity $s(M)$ is the number of non-zero entries in $M$. 
We say that a matrix $M\in \F^{n\times n}$ of rank $k$ \emph{contains} an \emph{$s$-sparse column (row) basis}, if $M$ \emph{contains} a column (row) basis 
(i.e., a set of $k$ linearly independent columns (rows)) with a total of at most $s$ non-zero entries. 

\begin{definition}[Minrank~\cite{BBJK06, LU07}]\footnote{In this paper we consider the directed version of minrank. Since the minrank of a directed graph does not exceed the minrank of its undirected counterpart, a lower bound for a directed random graph implies the same lower bound for an undirected random graph. The bound is tight for both directed and undirected random graphs (see Theorem~\ref{thm:tight}). 
}\label{def_mrk}
Let $G=(V,A)$ be a graph on $n=|V|$ vertices with the set of directed arcs $A$. A matrix $M \in \F^{n\times n}$ \emph{represents} 
$G$ if $M_{i,i}\neq0$ for every $i\in[n]$, and $M_{i,j}=0$ whenever $(i,j)\notin A$ and $i\neq j$. The minrank of $G$ over $\F$ is 
\[
\mrk(G)= \min_{M \text{ represents }G}  \rk(M) \; .
\]
\end{definition}

\

We say that two graphs \emph{differ at only one vertex} if they differ only in arcs leaving one vertex. Following~\cite{HHMS10,HL12}, to amplify the probability in Theorem~\ref{thm_mrk_gnp}, we shall use the following form of 
Azuma's inequality for the vertex exposure martingale.

\begin{lemma}[Corollary 7.2.2 and Theorem 7.2.3 in~\cite{AS15}]\label{thm:azuma}
Let $f(\cdot )$ be a function that maps directed graphs to $\mathbb{R}$.
If $f$ satisfies the inequality $|f(H)-f(H')|\le1$ whenever the graphs $H$ and $H'$ differ at only one vertex, 
then 
\[
\Pr[\left|f(\cG_{n,p})-\E[f(\cG_{n,p})]\right|>\lambda\sqrt{n-1}]<2e^{-\lambda^2/2} \; .
\]
\end{lemma}

\section{The Minrank of a Random Graph} \label{sec_minkank_gnp}

The following elementary linear-algebraic lemma shows that a matrix  $M \in \F^{n\times n}$ of rank $k$ is fully specified by $k$ linearly 
independent rows, $k$ linearly independent columns, and their $2k$ indices. In what follows, 
we denote by $\cM_{n,k}$ the set of matrices from $ \F^{n\times n}$ of rank $k$. 
\begin{lemma}[Row and column bases encode the entire matrix] \label{lem_basis_determines_matrix}
Let $M \in \cM_{n,k}$, and let $R = (\row_{i_1}(M),\ldots , \row_{i_k}(M)), C = (\col_{j_1}(M),\ldots , \col_{j_k}(M))$ be, 
respectively, a row basis and a column basis of $M$.
Then the mapping 
$\phi\colon \cM_{n,k} \to (\F^{1\times n})^k \times (\F^{n\times 1})^k \times [n]^{2k}$
defined as 
\[
\phi(M)=(R,C,i_1,\ldots,i_k,j_1,\ldots,j_k) \; ,
\]
 is a one-to-one mapping. 
\end{lemma}

\begin{proof}
We first claim that the intersection of $R$ and $C$ has full rank, i.e., that the submatrix $M'\in\F^{k\times k}$ obtained by taking rows $i_1,\ldots,i_k$ and columns $j_1,\ldots,j_k$
has rank $k$. This is a standard fact, see, e.g.,~\cite[p20, Section 0.7.6]{HornJohnson}. We include a proof for completeness. 
Assume for convenience that $(i_1,\ldots,i_k)=(1,\ldots, k)$ %
and $(j_1,\ldots,j_k)=(1,\ldots, k)$. %
Next, assume towards contradiction that $\rk(M')=\rk(\{\col_1(M'),\ldots , \col_{k}(M')\}) = k' < k$.
Since $C$ is a column basis of $M$, every column $\col_i(M)$ is a linear combination of vectors 
from $C$, and in particular, every $\col_i(M')$ is a linear combination of $\{\col_1(M'),\ldots , \col_{k}(M')\}$. 
Therefore, the $k\times n$ submatrix $M'' \coloneqq (\col^{\leq k}_1(M),\ldots , \col^{\leq k}_{n}(M))$ has rank $k'$. On the other hand, the $k$ rows of $M''\colon \row_1(M),\ldots , \row_k(M)$ were chosen to be linearly independent by construction. Thus, $\rk(M'')=k>k'$, which leads to a contradiction. 

In order to show that $\phi$ is one-to-one, we show that $R$ and $C$ (together with their indices) uniquely determine the remaining entries of $M$. 
We again assume for convenience that $(i_1,\ldots,i_k)=(1,\ldots, k)$ %
and $(j_1,\ldots,j_k)=(1,\ldots, k)$.
Consider any column vector $\col_i(M)$, $i\in [n]\setminus [k]$. By definition, 
$\col_i(M) = \sum_{t=1}^k \alpha_{i,t} \cdot  \col_t(M)$ for some coefficient vector 
$\alpha_i\coloneqq (\alpha_{i,1},\ldots , \alpha_{i,k})\in\F^{k\times1}$. Thus,
in order to completely specify all the entries of $\col_i(M)$, it suffices to determine the coefficient vector $\alpha_i$. But $M'$ has full rank, 
hence the equation $$ M' \alpha_i^T = \col^{\leq k}_i(M) $$ 
has a \emph{unique} solution. Therefore, the coefficient vector $\alpha_i$ is fully determined by $M'$ and $\col^{\leq k}_i(M)$.
Thus, the matrix $M$ can be uniquely recovered from $R, C$ and the indices $\{i_1,\ldots, i_k\}, \{j_1,\ldots,j_k\}$.
\end{proof}

The following corollary gives us an upper bound on the number of low-rank matrices that contain  sparse column and row bases.
In what follows, we denote by $\cM_{n,k,s}$ the set of matrices over $\F^{n \times n}$ of rank $k$ that contain an $s$-sparse 
row basis and an $s$-sparse column basis.
\begin{corollary}[Efficient encoding of sparse-base matrices]\label{cor_matrix_encoding}
 \[ \left| \cM_{n,k,s} \right| \leq (n\cdot|\F|)^{6s} \; . \]
\end{corollary}

\begin{proof}
Throughout the proof, we assume without loss of generality that $s \geq k$, as otherwise 
$\left| \cM_{n,k,s} \right|=0$ hence the inequality trivially holds. 
The function $\phi$ from Lemma~\ref{lem_basis_determines_matrix} maps matrices from $\cM_{n,k,s}$ to $(R,C,i_1,\ldots, i_k, j_1,\ldots, j_k)$, where $R$ and $C$ are $s$-sparse bases. Therefore, the total number of matrices in $\cM_{n,k,s}$ is bounded from above by 
\[
\left({\binom{kn}{s}} \cdot |\F|^{s}\right)^2 \cdot n^{2k} \leq \left( (n^2)^{s} \cdot |\F|^{s} \right)^2 \cdot n^{2k} \leq (n\cdot|\F|)^{6s} \; ,
\]
where the last inequality follows from $k\leq s$.
\end{proof}

Now we show that a matrix of low rank with nonzero entries on the main diagonal must contain many nonzero entries. To get some intuition on this, notice that a rank $1$ matrix 
with nonzero entries on the diagonal must be nonzero everywhere. Also notice that the assumption
on the diagonal is crucial -- low rank matrices in general can be very sparse. 

\begin{lemma}[Sparsity vs.\ Rank for matrices with non-zero diagonal] \label{lem_rank_vs_sparsity_diagonal}
For any matrix $M \in \F^{n\times n}$ with non-zero entries on the main diagonal (i.e., $M_{i,i}\neq0$ for all $i\in[n]$), it holds that
\[   s(M) \geq \frac{n^2}{4 \rk(M)} \; .  \]
\end{lemma}

\begin{proof}
Let $s$ denote $s(M)$. 
The average number of nonzero entries in a column of $M$ is $s/n$. Therefore, Markov's inequality implies that there are at least $n/2$ columns in $M$
\emph{each of which} has sparsity at most $2s/n$. 
Assume without loss of generality that these are the first $n/2$ columns of $M$.
Now pick a maximal set of linearly independent columns among these columns. 
We claim that the cardinality of this set is at least $n^2/(4s)$. 
Indeed, in any set of less than $n^2/(4s)$ columns, the number of coordinates that are nonzero in \emph{any} of the columns is less than
\[
    \frac{n^2}{4s} \cdot \frac{2s}{n} = \frac{n}{2}
\]
and therefore there exists a coordinate $i \in \{1,\ldots,n/2\}$ that is zero
in all those columns. As a result, the $i$th column, which by assumption has a nonzero $i$th coordinate, must be linearly independent of
all those columns, in contradiction to the maximality of the set. We therefore get that
\[
  \rk(M) \geq n^2/(4s) \; ,
\]
as desired.
\end{proof}

The last lemma we need is also the least trivial. 
In order to use Corollary~\ref{cor_matrix_encoding}, we would like
to show that any $n \times n$ matrix of rank $k$ has sparse row and column bases,
where by sparse we mean that their sparsity is roughly $k/n$ times that of the entire matrix.
If the number of nonzero entries in each row and column was roughly the same,
then this would be trivial, as we can take any maximal set of linearly independent
columns or rows. 
However, in general, this might be impossible to achieve. E.g., consider 
the $n \times n$ matrix whose first $k$ columns are chosen uniformly and 
the remaining $n-k$ columns are all zero. Then any column basis would 
have to contain all first $k$ columns (since they are linearly 
independent with high probability) and hence its sparsity is equal
to that of the entire matrix. 
Instead, what the lemma shows is that one can always choose a 
\emph{principal submatrix} with the desired property, i.e., that 
it contains sparse row and column bases, while at the same time having
relative rank that is at most that of the original matrix.

\begin{lemma}[Every matrix contains a principal submatrix of low relative-rank and sparse bases] \label{lem_k'_n'_n_k}
Let $M\in \cM_{n,k}$ be a matrix. %
There exists a principal submatrix $M'\in \cM_{n',k'}$ of $M$, such that 
$k'/n' \leq k/n$, and $M'$ contains a column basis and a row basis of sparsity at most 
\[
s(M')\cdot\frac{2k'}{n'}\; .  %
\]
\end{lemma}

Note that if $M$ contains a zero entry on the main diagonal, the lemma becomes trivial. Indeed, we can take $M'$ to be a $1\times 1$ principal submatrix formed by this zero entry. Thus, the lemma is only interesting for matrices $M$ without zero elements on the main diagonal (i.e., when every principal submatrix has rank greater than $0$).

\begin{proof}
We prove the statement of the lemma by induction on $n$. 
The base case $n=1$ holds trivially.

Now let $n>1$, and assume that the statement of the lemma is proven for every $m\times m$ matrix for $1\leq m < n$.
Let $s(i)$ be the number of nonzero entries in the $i$th column plus the number of non-zero entries in the $i$th row (note that a nonzero entry on the diagonal is counted twice).
Let also $s_{\max}=\max_{i}{s(i)}$. 
By applying the same permutation to the columns and rows of $M$ we can assume that $s(1)\leq s(2)\leq \cdots \leq s(n)$ holds. 

If for some $1\leq n'<n$, the leading principal submatrix $M'$ of dimensions $n'\times n'$ has rank at most $k'\leq n'k/n$, then we use the induction hypothesis for $M'$. This gives us a principal submatrix $M''$ of dimensions $n''\times n''$ and rank $k''$, such that $M''$ contains a column basis and a row basis of sparsity at most $s(M'')\cdot\frac{2k''}{n''}$. Also, by induction hypothesis $k''/n''\leq k'/n'\leq k/n$, which proves the lemma statement in this case.

Now we assume that for all $n'<n$, the rank of the leading principal submatrix of dimension $n'\times n'$ is greater than $n'k/n$. We prove that the lemma statement holds for $M'=M$ for a column basis, and an analogous proof gives the same result for a row basis.

For every $0\leq i\leq s_{\max}$, let $a_i=|\{j \,:\, s(j)=i\}|$.
Note that 
\begin{equation}\label{eq:suma}
\sum_{i=0}^{s_{\max}} a_i=n \; .
\end{equation}
Let us select a column basis of cardinality $k$ by greedily adding linearly independent vectors to the basis in non-decreasing order of $s(i)$. 
Let $k_i$ be the number of selected vectors $j$ with $s(j)=i$. 
Then 
\begin{equation}\label{eq:sumk}
\sum_{i=0}^{s_{\max}} k_i=k.
\end{equation}

Next, for any $0 \leq t<s_{\max}$,
consider the leading principal submatrix given by indices $i$ with $s(i)\leq t$. 
The rank of this matrix is at most $k'=\sum_{i=0}^t k_i$, and its dimensions are $n'\times n'$, where $n'=\sum_{i=0}^t a_i < n$. Thus by our assumption $k'/n'\geq k/n$,
or equivalently, 
\begin{equation}\label{eq:ineq_on_t}
\sum_{i=0}^tk_i\geq\frac{k}{n}\cdot \sum_{i=0}^ta_i \; .
\end{equation}

From~\eqref{eq:suma} and \eqref{eq:sumk}, 
\begin{equation}\label{eq:totalsum}
\sum_{i=0}^{s_{\max}}{k_i}=\frac{k}{n}\cdot\sum_{i=0}^{s_{\max}}{a_i} \; .
\end{equation}
Now, \eqref{eq:ineq_on_t} and~\eqref{eq:totalsum} imply that for all $0\leq t\leq s_{\max}$:\begin{equation}\label{eq:suffix}
\sum_{i=t}^{s_{\max}}{k_i} \leq \frac{k}{n}\cdot \sum_{i=t}^{s_{\max}}{a_i} \; .
\end{equation}
To finish the proof, notice that the sparsity of the constructed basis of $M$ 
is at most 
\begin{align*}
\sum_{i=1}^{s_{\max}}i\cdot k_i
=\sum_{t=1}^{s_{\max}}\sum_{i=t}^{s_{\max}}k_i
\stackrel{\eqref{eq:suffix}}{\leq}
\frac{k}{n}\cdot\sum_{t=1}^{s_{\max}}\sum_{i=t}^{s_{\max}}a_i
=\frac{k}{n}\cdot\sum_{i=1}^{s_{\max}}i\cdot a_i=s(M)\cdot \frac{2k}{n} \; .
\end{align*}
\end{proof}

Now we are ready to prove our main result  -- a lower bound on the minrank of a random graph.

\begin{theorem} \label{thm_mrk_gnp} 
\[
\Pr\left[\, \mrk(\cG_{n,p}) \geq \Omega\left(\frac{n\log(1/p)}{\log{\left(n|\F|/p\right)}}\right) \, \right ] \; \geq \; 1-e^{-\Omega\left(\frac{n\log^2{(1/p)}}{\log^2{\left(n|\F|/p\right)}}\right)} \; .
\]
\end{theorem}

\begin{proof}
Let us bound from above probability that a random graph $\cG_{n,p}$ has minrank at most 
$$k :=  \frac{n\log(1/p)}{C\log{\left(n|\F|/p\right)}},$$ for some constant $C$ to be chosen below.

Recall that by Lemma~\ref{lem_k'_n'_n_k}, every matrix of rank at most $k$ contains a principal submatrix $M'\in\cM_{n',k'}$ of sparsity 
$s' = s(M')$ with column and row bases of sparsity at most 
$$s'\cdot\frac{2k}{n},$$ where $k'/n'\leq k/n$. %
By Corollary~\ref{cor_matrix_encoding}, there are at most $(n'\cdot|\F|)^{6(2s'k/n)}$ such matrices $M'$, 
and (for any $s'$) there are $\binom{n}{n'}$ ways to choose a 
principal submatrix of size $n'$ in a matrix of size $n\times n$. 
Furthermore, recall that Lemma~\ref{lem_rank_vs_sparsity_diagonal} asserts that for every $n',k'$, 
\begin{equation}\label{eq_density_M_prime}
s'\geq \frac{n'^2}{4k'}.
\end{equation}
Finally, since $M'$ contains at least $s'-n'$ off-diagonal non-zero entries, $\cG_{n,p}$ contains it with probability at most 
$p^{s'-n'}$. We therefore have 

\begin{align}
\Pr\left[\mrk(\cG_{n,p}) \leq k\right]\nonumber\\
&\leq\sum_{k',n',s'}\Pr\left[\cG_{n,p} \text{ contains $M'\in\cM_{n',k'},$ } s(M')=s', s(\text {bases of }M') 
 \leq s'\cdot\frac{2k}{n}\right] \nonumber\\
&\leq\sum_{k',n',s'}  \binom{n}{n'} \cdot p^{s'-n'}  \cdot \left(n'\cdot|\F| \right)^{12s'k/n}\nonumber\\
&\leq\sum_{k',n',s'} 2^{n'\log{n}-s'\log(1/p)+n'\log(1/p)+(12s'k/n)\log{(n'|\F|)}} \label{mainbound}
\; ,
\end{align}
where all the summations are taken over $n', k'$, s.t. $k'/n' \leq  k/n$ and $s'\geq \frac{n'^2}{4k'}$, and the first inequality 
is again by Lemma~\ref{lem_k'_n'_n_k}. 
We now argue that for sufficiently large constant $C$, all positive terms in the exponent of \eqref{mainbound} are 
dominated by the magnitude of the negative term ($s'\log(1/p)$). Indeed: 
\begin{align*}
n'\log{n}+n'\log(1/p)+(12s'k/n)\log{(n'|\F|)}
&=n'\log{(n/p)}+(12s'k/n)\log{(n'|\F|)}\\
\leq(4s'k'/n')\log{(n/p)}+(12s'k/n)\log{(n|\F|)}
&\leq(16s'k/n)\log{(n|\F|/p)}
=(16s'/C)\log{(1/p)}
\; ,
\end{align*}
where the first inequality follows from~\eqref{eq_density_M_prime}, and the second one follows from $k'/n' \leq k/n$.

Thus, for $C>16$, 
\begin{align*}
\Pr\left[\mrk(\cG_{n,p}) \leq \frac{n\log(1/p)}{C\log{\left(n|\F|/p\right)}}\right] \leq
n^4\cdot 2^{-\Omega(s'\log(1/p))}\leq 
2^{-\Omega(\log(n))}.
\end{align*}
In particular, $\E\left[\mrk(\cG_{n,p})\right] \geq \frac{n\log(1/p)}{2C\log{\left(n|\F|/p\right)}}$. Furthermore, 
note that changing a single row (or column) of a matrix can change its minrank by at most $1$, hence 
the minrank of two graphs that differ in one vertex differs by at most $1$. We may thus 
 apply Lemma~\ref{thm:azuma} with 
$\lambda=\Theta\left(\frac{\sqrt{n}\log(1/p)}{\log{\left(n|\F|/p\right)}} \right)$ to obtain 
\[
\Pr\left[\, \mrk(\cG_{n,p}) \geq \Omega\left(\frac{n\log(1/p)}{\log{\left(n|\F|/p\right)}}\right) \, \right ] \; \geq \; 1-e^{-\Omega\left(\frac{n\log^2{(1/p)}}{\log^2{\left(n|\F|/p\right)}}\right)} \; .
\]
as desired. 
\end{proof}

\begin{corollary}\label{cor_main_thm_const_p}
For a  constant $0<p<1$ and a field $\F$ of size $|\F|<n^{O(1)}$,
\[
\Pr\left[\, \mrk(\cG_{n,p}) \geq \Omega(n/\log{n}) \, \right ] \; \geq \; 1-e^{-\Omega\left(n/\log^2{n}\right)} \; .
\]
\end{corollary}

\subsection{Tightness of Theorem~\ref{thm_mrk_gnp}}

In this section, we show that Theorem~\ref{thm_mrk_gnp} provides a tight bound for all values of $p$ bounded away from~$1$ 
(i.e., $p \leq 1-\Omega(1)$). (See also the end of the section for the regime of $p$ close to $1$.)
\begin{theorem}\label{thm:tight}
For any $p$ bounded away from~$1$, 
\[
\Pr\left[\, \mrk(\cG_{n,p}) = O\left(\frac{n\log(1/p)}{{\log{n}+\log(1/p)}}\right) \, \right ] \; \geq \; 1- e^{-\Omega\left(n\right)}\; .
\]
\end{theorem}
\begin{proof}
We can assume that $p > n^{-1/8}$ as otherwise the statement is trivial.

As we saw in the introduction, in the case of a clique (a graph with an arc between every pair of distinct vertices) it is enough to broadcast only one bit. 
This simple observation leads to the ``clique-covering'' upper bound: If a directed graph $G$ can be covered by $m$ cliques, then $\mrk(G)\leq m$~\cite{haemers1978upper, BBJK06, HL12}. Note that the minimal number of cliques needed to cover $G$ is exactly $\chi(\bar{G})$. Thus, we have the following upper bound: For any field $\F$ and any directed graph $G$, 
\begin{align}\label{thm:upperbound}
\mrk(G)\leq \chi(\bar{G}) \; .
\end{align}
Since the complement of $\cG_{n,p}$ is $\cG_{n,1-p}$, it follows from~\eqref{thm:upperbound}
that an upper bound on $\chi(\cG_{n,1-p})$ implies an upper bound on $\mrk(\cG_{n,p})$.

Let $\cG^-_{n,p}$ denote a random Erd{\H o}s-R{\'e}nyi %
\emph{undirected} graph on $n$ vertices, where each edge is drawn independently with probability $p$.
For constant $p$, the classical result of Bollob{\'a}s~\cite{bollobas1988chromatic} asserts that the 
chromatic number of an undirected random graph satisfies 
\begin{equation}\label{eq:chiconst}
\Pr\left[\chi(\cG^-_{n,1-p})\leq \frac{n\log{(1/p)}}{2\log{n}}\left(1+o(1)\right)\right] > 1-e^{-\Omega(n)}  \; .
\end{equation}
In fact, Pudl{\'a}k, R{\"o}dl, and Sgall~\cite{pudlak1997boolean} 
showed that~\eqref{eq:chiconst} holds for any $p>n^{-1/4}$.

Since we define the chromatic number of a directed graph to be the chromatic number of its undirected counterpart, $\chi(\cG_{n,1-p})=\chi(\cG^-_{n,1-p^2})$.
The bound~\eqref{eq:chiconst} depends on $p$ only logarithmically ($\log{(1/p)}$), thus, asymptotically the same bounds hold for the chromatic number of a random directed graph. 
\end{proof}

The lower bound of Theorem~\ref{thm_mrk_gnp} is also almost tight for the other extreme regime of $p=1-\eps$, where $\eps=o(1)$. {\L}uczak~\cite{luczak1991chromatic} proved that for $p=1-\Omega(1/n)$, 
\begin{equation}\label{eq:chilarge}
\Pr\left[\chi(\cG^-_{n,1-p})\leq \frac{n(1-p)}{2\log{n(1-p)}}\left(1+o(1)\right)\right] > 1- \left(n(1-p)\right)^{-\Omega(1)}  \; .
\end{equation}
When $p=1-\eps$, the upper bound~\eqref{eq:chilarge} matches the lower bound of Theorem~\ref{thm_mrk_gnp} for $\eps\geq n^{-1+\Omega(1)}$. For $\eps=O(n^{-1})$,~\eqref{eq:chilarge} gives an asymptotically tight upper bound of $O(1)$. Thus, we only have a gap between the lower bound of Theorem~\ref{thm_mrk_gnp} and  known upper bounds when $p=1-\eps$ and $\omega(1) \leq n\eps \leq n^{o(1)}$. 

\section*{Acknowledgements}
We would like to thank Ishay Haviv for his valuable comments on an earlier version of this work.

\bibliographystyle{alpha}
\bibliography{refs}

\end{document}